\documentclass[a4paper]{article}
\usepackage[all]{xy}\usepackage[latin1]{inputenc}        
\usepackage[dvips]{graphics,graphicx}
\usepackage{amsfonts,amssymb,amsmath,color,mathrsfs, amstext}
\usepackage{amsbsy, amsopn, amscd, amsxtra, amsthm,authblk}
\usepackage{enumerate,algorithmicx,algorithm}
\usepackage{algpseudocode}
\usepackage{upref}
\usepackage{geometry}
\usepackage[displaymath]{lineno}
\usepackage{float}
\usepackage{bm}
\usepackage{caption}
\usepackage{yhmath}
\usepackage{booktabs}
\usepackage{subcaption}
\usepackage{multirow}
\usepackage{makecell}
\usepackage[colorlinks,
            linkcolor=red,
            anchorcolor=red,
            citecolor=red
            ]{hyperref}

\usepackage{epstopdf}
\usepackage{tabularx}
\usepackage{threeparttable}
\usepackage{dcolumn}
\usepackage{epsfig}
\usepackage{array}
\usepackage{hyperref}
\usepackage{setspace}
\numberwithin{equation}{section}

\DeclareMathOperator{\erf}{erf}
\DeclareMathOperator{\erfc}{erfc}

\DeclareMathOperator{\var}{var}

\newcommand{\tcb}[1]{\textcolor{blue}{#1}}

\newcommand{\bbZ}{\mathbb{Z}}
\newcommand{\cO}{\mathcal{O}}

\newtheorem{theorem}{Theorem}[section]

\numberwithin{equation}{section}

\begin{document}

\title{Improved random batch Ewald method in molecular dynamics simulations}

\author[1]{Jiuyang Liang}
\author[1,2]{Zhenli Xu\thanks{xuzl@sjtu.edu.cn}}
\author[1]{Yue Zhao}
\affil[1]{School of Mathematical Sciences, Shanghai Jiao Tong University, Shanghai, 200240, P. R. China}
\affil[2]{MOE-LSC, CMA-Shanghai and Shanghai Center for Applied Mathematics, Shanghai Jiao Tong University, Shanghai 200240, P. R. China}
\date{\today}

\maketitle

\begin{abstract}
The random batch Ewald (RBE) is an efficient and accurate method for  molecular dynamics (MD) simulations of
physical systems at the nano-/micro- scale. The method shows great potential to solve the computational bottleneck of long-range interactions, motivating a necessity to accelerating short-range components of the non-bonded interactions for a further speedup of MD simulations.
In this work, we present an improved RBE method for the non-bonding interactions by introducing the random batch idea
to constructing neighbor lists for the treatment of both the short-range part of the Ewald splitting and the Lennard-Jones potential.
The efficiency of the novel neighbor list algorithm owes to the stochastic minibatch strategy which can significantly reduce the total
number of neighbors. We obtan the error estimate and convergence by theoretical analysis and implement the improved RBE method in the LAMMPS package.  Benchmark simulations are performed to demonstrate the accuracy and stability of the algorithm. Numerical tests on computer performance by conducting large-scaled MD simulations for systems including up to 0.1 billion water molecules, run on massive cluster with up to 50 thousand CPU cores, demonstrating the attractive features such as the high parallel scalability and memory-saving of the method in comparison to the existing methods.

{\bf Key words}. Molecular dynamics simulations, electrostatics, random batch Ewald, random batch list

\end{abstract}

\section{Introduction}
Molecular dynamics (MD) furnishes a powerful tool for understanding equilibrium and dynamical properties of a broad range of systems at the molecular and atomic level \cite{karplus1990molecular,hollingsworth2018molecular,yamakov2002dislocation,brunger1987crystallographic,allen2017computer}, including physical, chemical, biological and materials sciences.
To obtain the trajectories of particles by solving the equations of motion, one essential element for an MD simulation is a knowledge of the inter-particle potential from which the force acting on each particle can be calculated.
The interaction force may vary from intramolecular forces to more complicated many-body forces between atoms and molecules, dominating the central processing unit (CPU) cost in comparison with other computational procedures such as thermostat and barostat.
Once the interaction force can be successfully calculated, the equations of motion by Newton's law are integrated to obtain the spatial position and temporal velocity of each atom at each time step.

Classical MD uses a molecular mechanics force field to model the inter-particle forces, which is a parameter set by fitting results of quantum mechanical calculations and, typically, to certain experimental measurements.
Generally, the force field model  \cite{brooks1983charmm} is composed of bond stretching, bending and torsional forces, and two non-bonded interactions including the van der Waals and electrostatic forces.
From computational point of view, the bonded interactions are less expensive as  the involving atoms are no more than three covalent bonds.
Whereas, the non-bonded forces are the computational bottleneck as the electrostatic forces must be calculated between all pairs of particles, and the van der Waals forces need to calculate the pair interactions less than some cutoff radius ($0.9\sim 1.5 nm$),
remaining a broad interest for algorithm development and optimization. Particularly, electrostatic interactions are ubiquitous in biomolecular and material systems such as DNA aggregation \cite{burak2003onset}, protein folding/unfolding \cite{davis2016role,zhou2018electrostatic,ghosh2022curious}, the form of surface pattern \cite{schiel2021molecular}, ion adsorption\cite{frigini2022effect}, and polyelectrolyte complexation \cite{lund2013charge}. An efficient and accurate electrostatic solver plays essential role for the simulations of these systems.

Various fast methods have been developed for the non-bonded interactions. It is noted that the Lennard-Jones (LJ) model is often used for the van der Waals force and it can be calculated with $\mathcal{O}(N)$  complexity due to the cutoff scheme  \cite{frenkel2001understanding}
which truncates the interaction potential between a pair of particles at a radial cutoff distance and ignores the pairs of larger distances.
Techniques such as the Verlet list method \cite{verlet1967computer} and the linked cell list method \cite{quentrec1973new} are often employed.
In mainstream MD packages \cite{thompson2021lammps,abraham2015gromacs}, electrostatic interactions are calculated by the Ewald-type lattice summation where the long-ranged smooth part is treated on uniform mesh via fast Fourier transform\cite{DYP:JCP:1993,HE::1988,EPB+:JCP:1995} (FFT) with $\mathcal{O}(N\log N)$ complexity. The remaining part is short ranged and can be calculated by the cutoff scheme similar to the LJ interactions.
It is mentioned that there are other kinds of powerful and linear-scaling method for fast evaluating electrostatic interactions, including fast multipole methods\cite{RN11,RN12}, multigrids and Maxwell-equation molecular dynamics\cite{maggs2002local}, where the fast multipole methods are particulary efficient  for systems for  large scale problems and for the case of inhomogeneous particle distribution \cite{kohnke2020gpu,LIANG2022108332}.
These fundamental methods for computing forces are all required to combine with modern distributed architectures, i.e., the well-known 3D domain decomposition\cite{frenkel2001understanding,plimpton1995fast,begau2015adaptive}. The advances on domain decomposition, building neighbor lists, and time integration are contributed to the latest topics of algorithm development.

The calculation of long-ranged electrostatic interactions require intensive communications between cores, which significantly reduce the parallel efficiency for large scale systems \cite{allen2017computer,RN18,RN70,RN68}. Besides, the proportion of the CPU cost on short-range interactions also tends to play an essential
contribution due to the balance strategy \cite{thompson2021lammps,abraham2015gromacs} between short-range and long-range interactions.
A crucial observation \cite{leiserson2020there,Sherry2021Proceeding} is that there are decreasing marginal returns to algorithmic innovation, because the easy-to-catch innovations have already been ``fished-out" \cite{kortum1997research} and the remaining is more difficult to find or provides smaller gains.
For MD simulations, this marginal effect also arises, where the bottleneck on communication latency for the calculation of non-bonding forces has no transformative advances for several decades \cite{RN18,DEShaw2020JCP}.

In recent years, stochastic algorithms emerge gradually and furnish an important bridge linking traditional methods and modern massive high-performance computing, lifting both efficiency and scalability\cite{liang2022superscalability,lu202186}. For electrostatic interactions, the random batch Ewald (RBE) method proposed recently \cite{jin2021random,liang2022superscalability}, has presented its tremendous potential to overcome the scalability issues for deterministic algorithms such as particle-mesh Ewald. The RBE method is based on the Ewald summation. It avoids the use of the FFT
by employing a random mini-batch importance sampling strategy on the Fourier components to approximate the force and
pressure contributions from the long-range part. Since its superior speedup in the Fourier space, it is pointed out \cite{liang2022superscalability,liang2021random1} that the short-range part becomes the bottleneck of the MD simulations,
and a further acceleration requires a high-efficient treatment of the short-range contributions of the non-bonded force,
including the LJ and the real part of the Ewald splitting. In this paper we propose an improved RBE (IRBE) method
employing the random batch list (RBL) scheme \cite{liang2021random} to accelerate the short-range calculation under the RBE
framework. Basically, the RBL does the random batch in the real space, resulting in a significant reduction in the number
of neighbors for each particle.
The RBL is an extension of the original random batch method \cite{jin2020random} by
dividing the region into a core-shell structure and constructing the minibatch for particles in the shell region.
This RBL idea works well for the LJ potential \cite{liang2021random}, and in this paper we demonstrate its extension 
for the short-range part of Coulomb interaction is effective and theoretically obtain the error estimate.
The IRBE method is implemented in the LAMMPS package, and a systematic investigation was conducted on both the primitive-model eleoctrlyte  and 
and all-atom  systems including the bulk water and a micro-phase separated LiTFSI ionic liquid.
Our numerical results demonstrate the attractive performance of the IRBE method, showing its great promise in saving computer resources of both CPU and memory costs, thus furnishing a useful tool to address the above-mentioned marginal effect.

\section{Methods}
\subsection{Classical non-bonded algorithms}
Before presenting the IRBE method, we briefly review classical methods for evaluating non-bonded interactions in a 3D periodic system,
including the Ewald and neighbor list methods.

Consider a system of $N$ atoms with charge $q_i$ at position $\bm{r}_i$ for $i=1,\cdots,N$. These atoms
are within a cubic simulation box $\Omega$ with side length $L$ and volume $V=L^3$, with a periodic boundary condition
specified to mimic a bulk environment. The charges in the system obey the neutrality condition
$\sum_j q_i=0.$
Under these conditions, the non-bonded potential energy due to the LJ and electrostatic interactions can be written
as the following lattice summation,
\begin{equation}\label{eq:Utot}
U=\frac{1}{2}\sum_{\bm{n}}{}'\sum_{i,j=1}^N \left\{q_iq_j \frac{1}{|\bm{r}_{ij}+\bm{n}L|}+4\epsilon_{ij} \left[\left(\frac{\sigma_{ij}}{|\bm{r}_{ij}+\bm{n}L|}\right)^{12}-\left(\frac{\sigma_{ij}}{|\bm{r}_{ij}+\bm{n}L|}\right)^6\right]\right\},
\end{equation}
where $\bm{r}_{ij}:=\bm{r}_j-\bm{r}_i$, $\bm{n}$ runs over all 3D vector with integer components, and $\epsilon_{ij}$ and $\sigma_{ij}$ are the depth of the LJ potential well and the distance where the potential changes sign between $i$th and $j$th particles, respectively.
The prime in Eq.\eqref{eq:Utot} is understood that the singular term when $i=j$ and $\bm{n}=(0,0,0)$ should be excluded.
The difficulties of calculating Eq.\eqref{eq:Utot} is twofold.
First, the long-range nature of the Coulomb potential leads to the conditionally convergent series, thus the direct truncation will produce artifacts\cite{frenkel2001understanding} and should be avoided. Second, the short-range part requires expensive $\mathcal{O}(N^2)$ cost for directly searching neighbors within a given radius at each time step.

Ewald summation\cite{ewald1921berechnung} addresses the first problem by a splitting strategy which decomposes the Coulomb kernel into a sum of two components,
\begin{equation}\label{eq::ewaldsplitting}
\frac{1}{r}=\frac{\erf(\sqrt{\alpha}r)}{r}+\frac{\erfc(\sqrt{\alpha}r)}{r}.
\end{equation}
Here $\alpha$ is a positive parameter controlling the decay rate of real space and Fourier space series,
$\erf(x):=(2/\sqrt{\pi})\int_0^x e^{-u^2}du$ is the error function,
and $\erfc(x)=1-\erf(x)$ denotes the error complementary function. After the Ewald splitting, the first component in Eq.\eqref{eq::ewaldsplitting}
becomes long ranged and smooth, thus can be handled in the Fourier space via the so-called Fourier transform,
$
\hat{f}(\bm{k})=\int_{\Omega} f(\bm{r})e^{-i\bm{k}\cdot\bm{r}}d\bm{r},
$
where $\bm{k}=2\pi\bm{m}/L$ with $\bm{m}\in\bbZ^3$. The second component in Eq.\eqref{eq::ewaldsplitting} is singular but short ranged, and thus can be truncated at a certain cutoff radii $r_\mathrm{c}$. Denote $k_c$ and $r_{lj}$ as the cutoff of the Fourier space and the cutoff of the LJ potential, respectively. The non-bonded potential energy $U$ by Eq.\eqref{eq:Utot} can be rewritten as a sum of the following four contributions:
\begin{equation}\label{eq:uuu}
\begin{aligned}
&U^{short}=\frac{1}{2}\sum_{|\bm{r}_{ij}+\bm{n}L|\le r_\mathrm{c}}{}'~q_iq_j\frac{\erfc(\sqrt{\alpha}|\bm{r}_{ij}+\bm{n}L|)}{|\bm{r}_{ij}+\bm{n}L|},\\
&U^{lj}=\frac{1}{2}\sum_{|\bm{r}_{ij}+\bm{n}L|\le r_{lj}}{}'~4\epsilon_{ij} \left[\left(\frac{\sigma_{ij}}{|\bm{r}_{ij}+\bm{n}L|}\right)^{12}-\left(\frac{\sigma_{ij}}{|\bm{r}_{ij}+\bm{n}L|}\right)^6\right],\\
&U^{long}=\frac{2\pi}{V}\sum_{\substack{\bm{k}\neq \bm{0}\\|\bm{k}|\leq k_c}}\frac{1}{|\bm{k}|^2}|\rho(\bm{k})|^2
e^{-|\bm{k}|^2/4\alpha},\\
&U^{corr}=-\sqrt{\frac{\alpha}{\pi}}\sum_{i=1}^N q_i^2,
\end{aligned}
\end{equation}
where $U^{short}$ and $U^{long}$ are short- and long-range parts of the electrostatic interactions, $U^{lj}$ is the LJ interaction, and $U^{corr}$ is the correction term due to the self energy. The structure factor $\rho(\bm{k})$ is the conjugate of the
Fourier transform of the charge density, defined as
\begin{equation}
\rho(\bm{k}):=\sum_{i=1}^N q_i e^{i\bm{k}\cdot\bm{r}_i}.
\end{equation}

By proper choice of the parameters ($\alpha$ and the real-space and Fourier-space cutoffs $r_\mathrm{c}$ and $k_c$), the computational complexity for both $U^{long}$ and $U^{short}$ is optimized to $\mathcal{O}(N^{3/2})$\cite{RN28,frenkel2001understanding}. Moreover, the FFT is often employed to further speed up the evaluation of $U^{long}$ such that the cutoff radius $r_\mathrm{c}$ can be much smaller, resulting in the core algorithms for mainstream software, including the particle-particle particle-mesh Ewald\cite{HE::1988} (PPPM), particle mesh Ewald\cite{DYP:JCP:1993}, and smooth particle mesh Ewald\cite{EPB+:JCP:1995} algorithms. The final computational complexity of $\mathcal{O}(N\log N)$ can be achieved through these methods for
periodic systems.

With the FFT acceleration on $U^{long}$, the cutoff radius $r_\mathrm{c}$ is often set to be the same as $r_{lj}$ so that $U^{short}$ and $U^{lj}$ can use the same neighbor list. Then, to handle the short-range part efficiently, a Verlet-style neighbor list\cite{verlet1967computer} is often created.
This neighbor list enumerates all pairs of atoms with separation less than a cutoff distance. The building of neighbor list involves a stencil of bins to check for possible neighbors, a procedure for binning atoms, and a loop on the resulting stencil to assemble the neighbor list according to some cutoff criterions. The Verlet list\cite{verlet1967computer} also introduces an additional larger cutoff radius to reduce the frequency of neighbor list establishment.
Building such a local neighbor list is in linear time, and the complexity of evaluating $U^{short}+U^{lj}$ also achieves linear scaling.
The prefactor is related to both the cutoff radii and the bin size.
A significant problem of classical neighbor list method is that the neighbor list typically consume the most memory of any data structure in the mainstream MD software\cite{thompson2021lammps,abraham2015gromacs}.
This is mainly because the average number of particles within the cutoff radius can be large for heterogeneous systems due to the use of a big radius.
As an example, the bulk water system is often simulated with cutoff radius $r_\mathrm{c}=1.2nm$ with which the average number
of neighbors is $\sim 700$.

A lot of techniques have been developed to optimize the Verlet list method, aiming at the reduction in both the memory access and computational cost.
Many of them are devoted to the improvement of the linked cell list method\cite{yao2004improved,welling2011efficiency} employed for the Verlet list construction.
There are also attempts on providing scalable algorithms that work well on a modern computer architecture,
including multi-core CPU, graphics processing unit, or even more in-depth coding technologies such at the single-instruction multiple
data (SIMD) vectorization instructions \cite{pall2013flexible,howard2016efficient,thompson2021lammps,tchipev2019twetris}.

Finally, the atomic non-bonded force $\bm{F}_i^{tot}$ is obtained as the negative gradient of the energy function $U$, which is given by,
\begin{equation}\label{eq:fffr}
\begin{aligned}
&\bm{F}_i^{short}=-q_i\sum_{r_{ij}\le r_\mathrm{c}}{}'~q_j \left[ \frac{\erfc(\sqrt{\alpha}r_{ij})}{r_{ij}^2} +\frac{2\sqrt{\alpha}e^{-\alpha r_{ij}^2}}{\sqrt{\pi}r_{ij}} \right] \frac{\bm{r}_{ij} }{r_{ij}},\\
&\bm{F}_i^{lj}=-\sum_{r_{ij}\le r_\mathrm{c}}{}'~\epsilon_{ij}\left(\dfrac{48\sigma_{ij}^{12}}{r_{ij}^{14}}-\dfrac{24\sigma_{ij}^{6}}{r_{ij}^{8}}\right)\bm{r}_{ij},\\
&\bm{F}_i^{long}=-\sum_{\bm{k}\neq \bm{0}}\frac{4\pi q_i \bm{k}}{V |\bm{k}|^2}
e^{-|\bm{k}|^2/4\alpha}\mathrm{Im}(e^{-i\bm{k}\cdot\bm{r}_i}\rho(\bm{k})),
\end{aligned}
\end{equation}
where $r_{ij}=|\bm{r}_{ij}|$, and $j$ runs over all particles in the central box and periodic images.
It is noted that $\bm{F}_i^{short}$ and $\bm{F}_i^{lj}$ are singular but short-range, and $\bm{F}_i^{long}$ is smooth and long-range, corresponding to the terms in Eq.\eqref{eq:uuu}. The correction term in \eqref{eq:uuu} is constant for given $\alpha$ and makes no contribution to the force.

As was pointed out\cite{pennycook2013exploring,tchipev2019twetris}, the enormous memory consumption of neighbor list for evaluating short-range interactions and the intensive communication of FFT for evaluating long-range interactions are two bottlenecks of MD simulations, limiting both the system scale and the time scale.
In the next section, we attempt to address these bottlenecks by developing the IRBE method.

\subsection{Random batch importance sampling}\label{sec::rbe}
The RBE is a fast algorithm for calculating the long-range force $\bm{F}_i^{long}$. Its idea is based on an observation that the Gaussian factor $e^{-|\bm{k}|^2/4\alpha}$ in Eq.\eqref{eq:fffr} can be normalized as a discrete probability distribution\cite{jin2021random}.
Consider the following identity,
\begin{equation}\label{eq:S}
S:=\sum_{\bm{k}\neq \bm{0}}e^{-\bm{k}^2/(4\alpha)}=H^3-1,
\end{equation}
where $H$ can be represented as follows by the Poisson summation formula,
\begin{equation}\label{eq:H}
H:=\sum_{m\in \bbZ}e^{-\pi^2 m^2/(\alpha L^2)}
=\sqrt{\frac{\alpha L^2}{\pi}}\sum\limits_{m\in \bbZ}e^{-\alpha m^2L^2}.
\end{equation}
In our applications, $\alpha L^2\gg 1$ holds and one can truncate $m$ at $\pm 1$, resulting in an efficient approximation in the rightmost of Eq.\eqref{eq:H}. One can regard the series summation as a functional expectation over the probability distribution
\begin{equation}\label{eq:probexpression}
\mathcal{P}_{\bm{k}}:=S^{-1}e^{-|\bm{k}|^2/(4\alpha)},\quad \bm{k}\neq\bm{0}.
\end{equation}

Since the Gaussian distribution is separable in multiple dimensions, $\mathcal{P}_{\bm{k}}$ can be sampled independently in each axis
with $\bm{k}=\bm{0}$ being skipped.
For more details, one samples from 1D Gaussian distribution $x^*\sim \mathcal{N}(0, \alpha L^2/(2\pi^2))$ and set $m^{\ast}=\mathrm{round}(x^{\ast})$, the acceptance rate is high via the Metropolis-Hasting sampling algorithm\cite{frenkel2001understanding}.
This procedure is independently repeated for three times.
A highly efficient parallel strategy for the importance sampling under NVT/NPT ensemble is also developed\cite{liang2022superscalability}.
The MD simulations can then be performed via the random mini-batch approach with this importance sampling strategy.
Specially, one can approximate the long-range force $\bm{F}_i^{long}$ in Eq.\eqref{eq:fffr} by the following random variable
\begin{equation}\label{eq:RBEFF}
\bm{F}_i^{long,*}:=-\sum\limits_{\ell=1}^{p} \frac{S}{p}\frac{4\pi \bm{k}_\ell q_i}{V |\bm{k}_\ell|^2}\mathrm{Im}(e^{-i\bm{k}_\ell\cdot\bm{r}_i}\rho(\bm{k}_\ell)),
\end{equation}
where $\bm{k}_{\ell}$, $1\le \ell \le p$, are the sampled frequencies. In the RBE method\cite{jin2021random}, the choice of parameter $\alpha$ shares the same strategy as in the PPPM method\cite{deserno1998mesh}, such that the computational complexity for the short-range interactions is $\mathcal{O}(N)$. If we choose $p$ as an $\mathcal{O}(1)$ constant, the complexity for approximating $\bm{F}_i^{long}$ via Eq.\eqref{eq:RBEFF} is also linear, namely, with complexity $\mathcal{O}(pN)$.

\subsection{Improved RBE method}\label{sec::rbl}

The RBE has shown its superscalibility in large-scale all-atom simulations \cite{liang2022superscalability},
resulting in that the evaluation of short-range non-bonded interactions becomes the most time-consuming part\cite{liang2021random1}.
To reduce the calculation for short-range interactions and save CPU memory in the original RBE method,
we introduce the RBL algorithm, first proposed for pure LJ fluid systems\cite{liang2021random}.


The RBL idea is based on a neighbor-splitting strategy. Let $r_\mathrm{c}$ be the cutoff radius for calculating both
$\bm{F}_i^{short}$ and $\bm{F}_{i}^{lj}$; typically, $r_\mathrm{c}=1.2 nm$ for all-atom simulations thus there are hundreds of neighbors
for each particle should be stored in the neighbor list.
The RBL method introduces the second radius $r_{\eta}<r_\mathrm{c}$ (e.g., $r_\eta=0.6 nm$) such that two-level core-shell structured neighbor lists are constructed around each particle. Direct summation is used for neighbors within the core region ($r_{ij}<r_\eta$), whereas a small number of particles from the shell zone are randomly chosen into a batch and other neighbors in the zone are ignored.
The central particle then interacts with the batch particles with a rescaled strength.
This idea is demonstrated that the average interacting neighbors of each particle are significantly reduced, and the accuracy is maintained by resolving the kernel singularity issue for the LJ systems\cite{liang2021random}.

Let $\bm{\mathcal{F}}_i=\bm{F}_i^{short}+\bm{F}_{i}^{lj}$ be the short-range non-bonded force on the $i$th particle.
Let $\mathcal{C}(i)=\{j\neq i: r_{ij}\leq r_\eta\}$ and $\mathcal{S}(i)=\{j: r_\eta<r_{ij}\leq r_\mathrm{c}\}$ be the neighbor sets of
the $i$th particle in the core and shell regions, respectively.
We now decompose it into the contributions from the core and the shell regions, separately,
$\bm{\mathcal{F}}_i=\bm{\mathcal{F}}_i^{core}+\bm{\mathcal{F}}_i^{shell}$, where
$\bm{\mathcal{F}}_i^{core}=\sum_{j\in \mathcal{C}(i)} \bm{f}_{ij}$,
and $\bm{\mathcal{F}}_i^{shell}=\sum_{j\in \mathcal{S}(i)} \bm{f}_{ij}$,
and the force contribution due to particle $j$ is,
\begin{equation}\label{eq::short1}
\bm{f}_{ij}=-q_iq_j \left[ \frac{\erfc(\sqrt{\alpha}r_{ij})}{r_{ij}^2} +\frac{2\sqrt{\alpha}e^{-\alpha r_{ij}^2}}{\sqrt{\pi}r_{ij}} \right] \frac{\bm{r}_{ij}}{r_{ij}}
-\epsilon_{ij}\left(\dfrac{48\sigma_{ij}^{12}}{r_{ij}^{14}}-\dfrac{24\sigma_{ij}^{6}}{r_{ij}^{8}}\right)\bm{r}_{ij}.
\end{equation}

Let $\mathcal{B}(i)$ be the batch of $\widetilde{p}$ particles randomly chosen from set $\mathcal{S}(i)$.
By following the random batch method \cite{jin2020random}, the force $\mathcal{F}_i^{shell}$ is approximated by
\begin{equation}\label{eq:3}
\bm{\mathcal{F}}_i^{shell,*}=\dfrac{N_{\mathcal{S}}}{\widetilde{p}}\sum_{j\in\mathcal{B}(i)}~\bm{f}_{ij},
\end{equation}
where $N_{\mathcal{S}}$ is the size of set $\mathcal{S}(i)$. It can be proved that the approximation is an unbiased estimate
of the exact force from those particles in $\mathcal{S}(i)$. Let $\bm{\mathcal{F}}^{corr}=\sum_i\left(\bm{\mathcal{F}}_i^{core}+\bm{\mathcal{F}}_i^{shell,*}\right)/N$
be the average net force
on each particle, which is a random variable with zero expectation and bounded variance.
One then obtains the stochastic approximation for short-range force, expressed by
\begin{equation}\label{eq:short1}
\bm{\mathcal{F}}_i\approx \bm{\mathcal{F}}_i^{core}+\bm{\mathcal{F}}_i^{shell,*}-\bm{\mathcal{F}}^{corr}.
\end{equation}
It is noted that the subtraction of the net force ensures the conservation of the total momentum in the system.
We shall also remark that the random batch idea leads to a significant reduction of neighbors, as analyzed below.

Now by introducing the RBL for the real space to the RBE, we develop an integrated stochastic approximation
of the non-bonded force, namely the IRBE method. The error estimate of this approximation is given in Error analysis Section\ref{sec::analysis}.
We remark that we use two batch sizes $\widetilde{p}$ and $p$ for the short-range force and the long-range force, respectively.
These two sizes are determined by the specific systems and can be different. They are both $\mathcal{O}(1)$ constants.
Another remark is that the IRBE method brings in an additional variance in the force term, leading to the numerical heating effect.
Therefore, at the moment, this method is not suitable for long time simulations under the microcanonical ensembles,
similar to the previous work \cite{jin2021random,liang2021random}. One shall develop an appropriate symplectic scheme
for the time integration of the equations of motion, which remains an open problem for our stochastic algorithms. Fortunately,
it is practical for the NVT and NPT ensembles with the use of thermostats and barostats.


In practice, to improve the efficiency, the simulation box is divided into uniform cells of edge $r_{\eta}$,
and the particle list in each cell is built.
For a given particle $i$, the core list is constructed from the nearest $27$ cells containing particles with
distance less than $r_{\eta}$.
Then, a stencil of cells, i.e., a combination of all neighboring particles of $i$ into uniformly sized cells of width $r_{\eta}$, is used for constructing the shell list.
In other words, the stencil is a larger cubic box comprised of $\left(2\lfloor r_\mathrm{c}/r_{\eta}\rfloor+3\right)^3$ cells where particle $i$ is located at the central cell.
It is feasible to apply the RBL method directly within this large stencil, avoiding the filter step for all particles.
Then the core and shell neighbor sets are constructed and the corresponding forces $\bm{\mathcal{F}}_i^{core}$ and $\bm{\mathcal{F}}_i^{shell,*}$ can be calculated.
Two things need to be remarked.
First, in practice, the edge of the cell can be slightly larger than $r_{\eta}$ so that updating the neighbor list is not required at every step\cite{thompson2021lammps}.
Second, depending on the interatomic potentials, multiple neighbor lists and stencils with different attributes may be needed.
One example is a solvated colloidal system with large
colloidal particles where colloid/colloid, colloid/solvent, and solvent/solvent interaction cutoffs can be dramatically different\cite{in2008accurate,howard2016efficient}.
The procedure of the IRBE method is summarized in Algorithm \ref{algo:IRBE}.

It is stressed that the storage consumption and computational complexity of the RBL are all relatively small, in comparison to
the classical direct truncation method with the Verlet list approach.
Without loss of generality, suppose that the particles are uniformly distributed with $\rho$ be the average particle density.
The calculation complexity and the CPU memory usage for storing the neighbor list per particle in the classical method is $\mathcal{O}(4\pi r_\mathrm{c}^3\rho/3)$, whereas the RBL reduce these cost to $\mathcal{O}(4\pi r_{\eta}^3\rho/3+\widetilde{p})$.
If one safely adopts $r_\mathrm{c}/r_{\eta} = 2.5$ and an appropriate $\widetilde{p}$, both the storage saving and the speedup
have about an order of magnitude improvement for the short-range interactions.
Regarding the Fourier space, the computational complexity of the RBE method is only $\cO(pN)$, where $p$ is often chosen as a few hundred, and it avoids the multiple, massive global communications of the FFT calculation.
These will be further discussed from numerical results in Results and discussion Section \ref{sec::NumericalResults}.

\begin{algorithm}[t]
	\caption{(Improved random batch Ewald)}\label{algo:IRBE}
	\begin{algorithmic}[1]
		\State Input initial data: number of particles $N$, volume $V$, temperature $T$,  simulation steps $N_{tot}$, and the initial particle information $\left\{\bm{r}_i,\bm{v}_i, q_i\right\}$. Choose parameters $\alpha$, $r_\mathrm{c}$, $r_{\eta}$, $k_c$, and batch sizes $p$ and $\widetilde{p}$
		\State Generate enough samples of $\bm{k}\sim e^{-|\bm{k}|^2/4\alpha}$ ($\bm{k} \neq 0$) by the Metropolis-Hastings procedure to form a sample set $\mathcal{K}$
		\For {$n$ in $1:N_{tot}$ }
		\State Divide the system into small cells of side length $r_{\eta}$, transmit particle information to different nodes, and create the core neighbor lists
		\State For each particle $i$, randomly choose $\widetilde{p}$ particles from its stencils if $N_{\mathcal{S}}>\widetilde{p}$, otherwise choose all particles
		\State Calculate short-range force $\bm{\mathcal{F}}_i^{core}$ with particles in the core region
		\State Calculate the stochastic force $\bm{\mathcal{F}}_i^{shell,*}$ in the shell zone by Eq.\eqref{eq:3} and the correction force
		\State Calculate the short-range forces $\bm{\mathcal{F}}_i^{*}\approx \bm{\mathcal{F}}_i^{core}+\bm{\mathcal{F}}_i^{shell,*}-\bm{\mathcal{F}}^{corr}$
		\State Choose $p$ samples from set $\mathcal{K}$ and calculate the long-range force $\bm{F}_i^{long,*}$ by Eq.\eqref{eq:RBEFF}
		\State Calculate the total non-bonding force $\bm{F}_{i}^{tot,*}=\bm{F}_i^{long,*}+\bm{\mathcal{F}}_i^{*}$
		\State Integrate Newton's equations with suitable integrate scheme and thermostat
		\EndFor
	\end{algorithmic}
\end{algorithm}

\subsection{Error analysis}\label{sec::analysis}
We conduct some analysis and discussions on the IRBE method to demonstrate its validity. Let $\bm{F}^{tot}_i$ be
the total non-bonded force on particle $i$, and $\bm{F}^{tot,*}_i$ is its approximation by the IRBE method.
Let $\bm{\chi}_{i}=\bm{F}^{tot}_i-\bm{F}^{tot,*}_i$ be the deviation of the approximate force.
It is obvious by following the proof in literature \cite{jin2021random,liang2021random} that the expectation of $\bm{\chi}_i$ is zero, i.e., $\bm{F}^{tot,*}_i$ is an unbiased estimator. The variance of the force approximation can be written as
\begin{equation}\label{eq::var}
\var[\bm{\chi}_{i}]=\var[\bm{\mathcal{F}}_i^{*}]+\var[\bm{F}_i^{long,*}]
\end{equation}
because the random mini-batches for approximating the short-range and the long-range force are mutually independent.

It is noted that the estimate of the long-range force has been obtained under the Debye--H\"uckel assumption \cite{jin2021random},
\begin{equation}\label{eq::long}
\var[\bm{F}_i^{long,*}]\leq \dfrac{CN^{4/3}}{pL^4}.
\end{equation}
Combining this result, we have Theorem \ref{thm:varshort} for the estimate for the total force.

\begin{theorem}\label{thm:varshort}
	Under the assumption of the Debye--H\"uckel theory for particle distribution, the variance of the difference between approximated and exact non-bonded force holds the following estimate:
	\begin{equation}\label{eq::varestimate}		
	\var[\bm{\chi}_{i}]\leq\dfrac{C(N_{\mathcal{S}}-\widetilde{p})}{\widetilde{p}}\rho
	\left[r_{\eta}^{-11}+\erfc\left(\sqrt{2\alpha}r_{\eta}\right)\right]+\dfrac{CN^{4/3}}{pL^4}.
	\end{equation}
\end{theorem}

\begin{proof}
	By following the results given in Refs.\cite{liang2021random,jin2021random}, the short-range part can be estimated by
	\begin{equation}\label{eq::short}
	\var[\bm{\mathcal{F}}_i^{*}]\leq\frac{C(N_{\mathcal{S}}-\widetilde{p})}{N^{2} \widetilde{p}}\left[\sum_{i \neq j} (N^2-2N) \left(\bm{f}_{ij}\right)^{2}+\sum_{k} \sum_{k \neq j} \left(\bm{f}_{ij}\right) ^{2}\right],
	\end{equation}
	where the minimum image conventions are already included in the force term \cite{frenkel2001understanding}.
	The short-range force $\bm{f}_{ij}$ is comprised of the short-range Ewald and the LJ forces.
	We consider these two parts separately, i.e., the variance of the short-range force is estimated by,
	\begin{equation}\label{eq::estimate1}
	\var[\bm{\mathcal{F}}_i^{*}]\leq \var[\bm{F}_i^{short,*}]+\var[\bm{F}_{i}^{lj,*}].
	\end{equation}
	
	We first consider the short-range Ewald force
	\begin{equation}\label{eq::Coulomb}
	\bm{f}^{short}_{i}(\bm{r})=-q_i q_j \left[ \frac{\erfc(\sqrt{\alpha}r)}{r^2} +\frac{2\sqrt{\alpha}e^{-\alpha r^2}}{\sqrt{\pi}r} \right] \frac{\bm{r}}{r}
	\end{equation}
	between particle $q_i$ and $q_j$ with the distance vector $\bm{r}$. When $r$ is large, one has
	\begin{equation}
	\frac{\erfc(\sqrt{\alpha}r)}{r^2}=\frac{2}{\sqrt{\pi}}\dfrac{\int_{\sqrt{\alpha}r}^{\infty}e^{-u^2}du}{r^2}\leq C\frac{e^{-\alpha r^2}}{r^3}=o\left(\frac{e^{-\alpha r^2}}{r}\right),
	\end{equation}
	where $C$ is constant depending on $\alpha$. Thus the upper bound on the variance of the approximation of the short-range Ewald force can be derived by estimating the contribution from the second term in Eq.\eqref{eq::Coulomb}
	\begin{equation}\label{eq::variance1}
	\var[\bm{F}_i^{short,*}]\leq\frac{(N_{\mathcal{S}}-\widetilde{p})}{\widetilde{p}}\int_{r_{\eta}}^{\infty}4\pi\rho r^2\left|\frac{e^{-\alpha r^2}}{r}\right|^2dr\leq \frac{C(N_{\mathcal{S}}-\widetilde{p})}{\widetilde{p}}\rho \erfc(\sqrt{2\alpha}r_{\eta}).
	\end{equation}
	Applying the same approach to the variance of the LJ force, one obtains
	\begin{equation}\label{eq::variance2}
	\var[\bm{F}_{i}^{lj,*}]\leq \dfrac{C(N_{\mathcal{S}}-\widetilde{p})}{\widetilde{p}}\rho r_{\eta}^{-11}.
	\end{equation}
	
	By combing Eqs.\eqref{eq::estimate1}, \eqref{eq::variance1},  \eqref{eq::variance2} with \eqref{eq::long},
	we obtain \eqref{eq::varestimate} and complete the proof.
\end{proof}

By Theorem \ref{thm:varshort}, we can safely suppose that $\var[\bm{\chi}_i]$ is bounded by a constant $\xi$ at the sense of infinite norm. 
In the following, we derive that the equilibrium distributions simulated by the true force and the random force has a small error proportional to $C\sqrt{\xi\Delta t}$ for a constant $C$ only depending on the total simulation time $t^*$, and $\Delta t$ being  the time step.

Let us consider the configurations produced by the underdamped Langevin dynamcs for the equations of motion,
\begin{equation}\label{SDE}
\begin{aligned}
&d\bm{r}_i=\bm{v}_idt,\\
&md\bm{v}_i=(\bm{F}_i^{tot}-\gamma \bm{v}_i)dt+\sqrt{2D}d\bm{W}_i,
\end{aligned}
\end{equation}
where $D=\gamma k_\text{B} T$ with $\gamma$ being the friction coefficient and $\{\bm{W}_i\}_{i=1}^N$ being the independently identially distributed Wiener processes.For convenience,
one defines $\bm{R}=(\bm{r}_1,\cdots,\bm{r}_N)'$, $\bm{V}=(\bm{v}_1,\cdots,\bm{v}_N)'$ and $\bm{Y}=[\bm{R};~\bm{M}\bm{V}]$ with  $\bm{M}$ being the diagonal mass matrix. These are quantities from the true force $\bm{F}_i^{tot}$.  $\bm{R}^*$, $\bm{V}^*$ and $\bm{Y}^*$ are denoted by the quantities from the approximate force $\bm{F}_i^{tot,*}$. The Langevin equations \eqref{SDE} can be collectively written as,
\begin{equation} 
d\bm{Y}= \bm{G} (\bm{Y})dt+ \bm{\mathcal{D}} d\bm{\mathcal{W}}, 
~\hbox{and}~d\bm{Y}^*= \bm{G} (\bm{Y}^*)dt+ \bm{\mathcal{D}} d\bm{\mathcal{W}} 
\end{equation}
where  $\bm{G}$ are the velocity and force terms,
$\bm{\mathcal{D}}=[\bm{0}, \sqrt{2D} \bm{I}]'$ and $\bm{\mathcal{W}}$ is
high-dimensional Wiener processes.	Consider the dynamics in a time step
$t\in[t_k,t_{k+1}]$ with $t_k=k\Delta t$. One defines the differences
$\bm{\delta Y}=\bm{Y}(t)-\bm{Y}^*(t_k)$ and $\bm{\delta G}=\bm{G}(\bm{Y}(t))-\bm{G}(\bm{Y}^*(t_k))$. The following equation holds,  
\begin{equation}\label{Ito1}
\dfrac{d}{dt}\mathbb{E}|\bm{\delta Y}(t)|^2=2\mathbb{E}[\bm{\delta Y}(t)]
\cdot \bm{\delta G}(t).
\end{equation}
By using the It\^o's formula   together with the bound $\xi$ for the force variance  by Theorem \ref{thm:varshort}, one can obtain the estimate (similar to literature\cite{jin2020random,li2020random}), 
\begin{equation}
\dfrac{d}{dt}\mathbb{E}|\bm{\delta Y}(t)|^2\leq C_1\left(\mathbb{E}|\bm{\delta Y}(t_k)|^2+(1+D^2)\Delta t^2\right)+2\xi\Delta t,
\end{equation}
for any $t\in[t_k,t_{k+1}]$, where $C_1$ is a constant.   
By the Gronwall's inequality, one reaches the strong convergence error,
\begin{equation}\label{eq::strongerr}
\sqrt{\mathbb{E}\left[\frac{1}{N}\sum_{i}\left(|\bm{r}_i-\bm{r}_i^*|^2+|\bm{v}_{i}-\bm{v}_{i}^*|^2\right)\right]}\leq C(t^{*})\sqrt{\xi\Delta t}.
\end{equation}

The convergence analysis \eqref{eq::strongerr} indicates that our method is valid for capturing the finite time dynamics in spite of the random batch approximation of the force. It is remarked that our derivation is based on the Langevin dynamics, but we actually take the Nos\'e-Hoover thermostat for the NVT ensemble in the simulations of next section. A rigorous derivation for the convergence of Nos\'e-Hoover thermostat by the RBE approximation is an open issue.

\section{Results and discussion}\label{sec::NumericalResults}

In this section, we perform numerical results on three typical systems: primitive-model electrolyte solutions, all-atom bulk water systems, and LiTFSI ionic liquids at the NVT ensemble in order to validate the accuracy and efficiency of the proposed IRBE method. The results of two different methods are also performed for
comparison. One is the PPPM method\cite{hockney2021computer} which is a classical method for popular MD packages. The other is the RBE method which does not introduce
the RBL for the short-range interactions. The calculations are conducted in the LAMMPS \cite{plimpton1995fast,thompson2021lammps} (version 29Oct2020) with the implementation of the RBE and IRBE, and on the ``Siyuan Mark-I'' cluster at Shanghai Jiao Tong University, which comprises $936$ nodes with 2 $\times$ Intel Xeon ICX Platinum 8358 CPU ($2.6$GHz, $32$ cores) and $512$ GB memory per node. Our implementation is optimized with distributed-memory parallelism via MPI and Intel 512-bit SIMD (AVX-512 architecture) instruction.
The communication and vectorization procedures are
described below in details.

A 3D domain decomposition strategy is coupled with for the calculation of real-space forces. The simulation box is spatially decomposed (partitioned) into non-overlapping subdomains which fill the box. An unique MPI rank or process is assigned to each subdomain, such that the computational tasks assigned to each rank are even out as far as possible. The 3D decomposition framework is of great important in modern MD software and has been attracted many interests \cite{thompson2021lammps,abraham2015gromacs,pall2013flexible}. We follow the procedure described in Ref.\cite{thompson2021lammps} for the dynamic load-balanced partitioning and the ghost-atom communication, whereas the construction of neighbor lists uses the strategy provided in the Improved RBE method Section \ref{sec::rbl} in this paper.

For the long-range force, owing to the random batch idea in the Fourier space, a serial importance sampling procedure and a global broadcast operation seems to be required at each MD step. We note that this cost can be significantly reduced by using parallel sampling strategy\cite{liang2022superscalability}. The samples and the positions of particles are packaged into 512-bit vectors when the structure factors $\rho(\bm{k})$ are evaluated using the local atoms of each MPI rank. Only one global operation, MPI$\_$Allreduce, is required for reducing $\rho(\bm{k})$. The approximated force $\bm{F}_{i}^{long,*}$ of each particle are then obtained from the structure factors.

\subsection{Electrolyte solution}
To demonstrate the performance of our approach, we first perform  MD simulations of simple 1:1 and 2:1 electrolytes in the canonical ensemble.
The electrolytes are described by the primitive model where ions are immersed in a continuum solvent
and represented as soft spheres that interact via a shifted-truncated LJ potential and electrostatic interactions.
The temperature is maintained by using a Langevin thermostat.
The simulation proceeds with velocity-Verlet scheme. In each simulation, we perform $5\times 10^5$ time steps for the equilibrium phase and $6\times 10^5$ time steps for the statistics.

The system includes $2560$ monovalent anions and $2560$ monovalent cations in 1:1 electrolyte, while $1280$ divalent cations and $2560$ monovalent anions in 2:1 electrolyte, and all quantities are provided in reduced units.
We fix the diameter of particles as $\sigma=1$, and the side length of the simulation box is $L=80\sigma$.
The MD time step is set as $\tau=0.002t_0$, where $t_0=\sigma\sqrt{m_0/k_B T}$ is the unit of time with the particle mass $m_0=1$.
The relaxation time is set to be $\gamma=1.0$ for the Langevin thermostat.

We first examine the accuracy by calculating the radial distribution function (RDF) between ions of different species and the mean-square displacement (MSD) of different kinds of ions.
The RDFs of atom pairs, denoted by $\textsl{g}_{++}$, $\textsl{g}_{-+}$ and so on, furnish the spatial arrangement of the electrolyte system and the MSD shows uniform linear motion in short time and diffusion behavior in long time of particles.
We use the PPPM method with $10^{-4}$ relative accuracy and cutoff $r_\mathrm{c}=15\sigma$ as the reference result.
In the IRBE, we set $r_{\eta}=5\sigma$ and batch size $\widetilde{p}=10$ for the short-range force, and $\alpha=0.03$ (the same as for the PPPM) and batch size $p=100$ (the same as for the RBE) for the long-range force.
The results are displayed in Fig.~\ref{fig:rdfmsd}, showing the three methods are almost overlapping in both the RDF and MSD curves,
demonstrating that the IRBE reproduces both structural and dynamical properties of the RBE and PPPM methods.
During the simulations, we compute the average number of neighbors per atom. There are $\sim 140$ neighbors are for both the PPPM and the RBE, whereas only $\sim 10$ of them are inside the core region, indicating that the IRBE significantly reduces the number of neighbors,
and thus it will be promising to save the memory and CPU costs, as is shown later on. In Fig.~\ref{fig:energy}, the potential energy per atom of every $100$ time steps are plotted by using different methods. The mean and standard deviation value are listed in Table.~\ref{tabl:ene}. The relative errors of the IRBE method have minor differences in comparison with the PPPM and the RBE.

\begin{figure*}[ht]	
	\centering
	\includegraphics[width=0.48\textwidth]{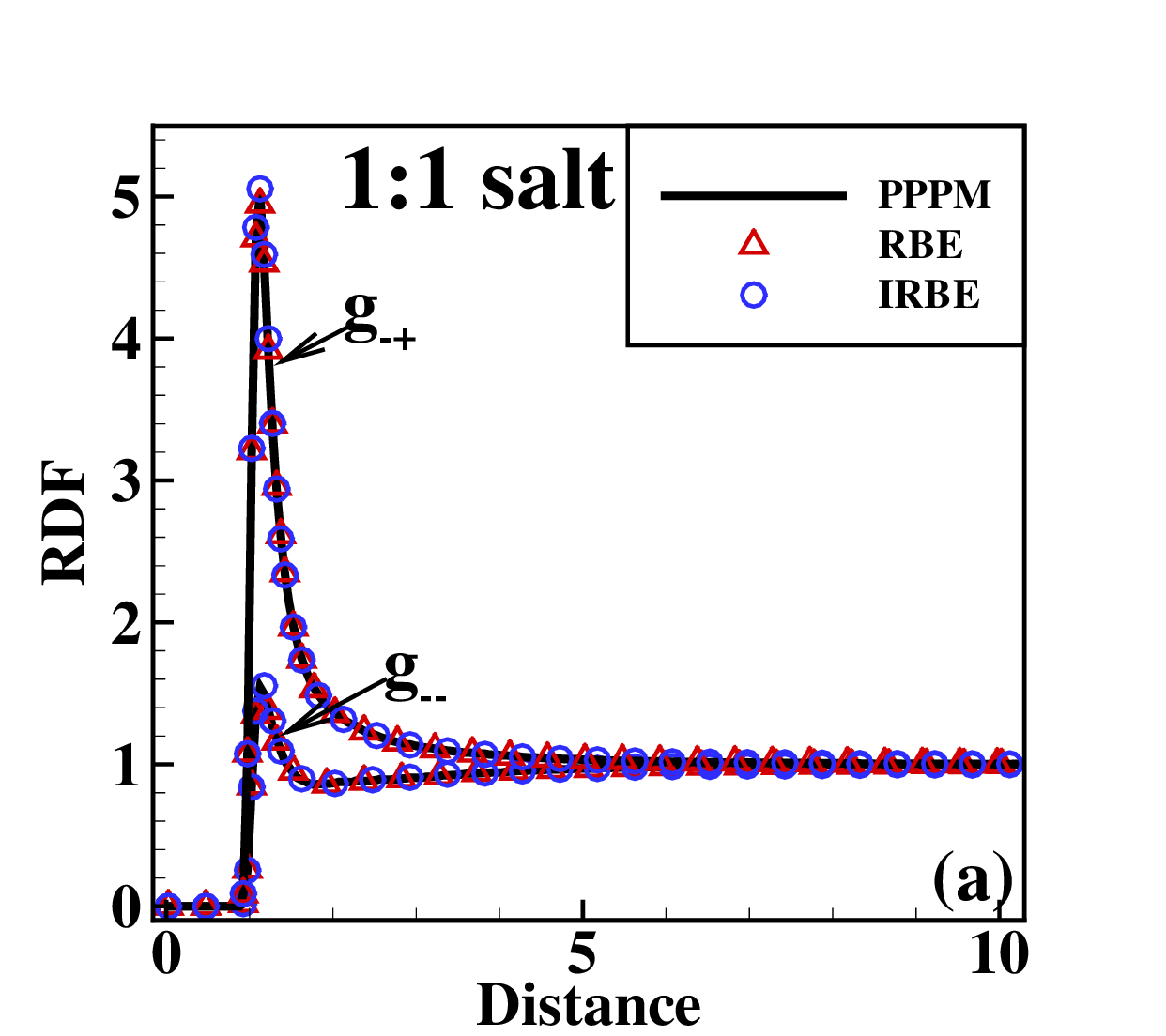}
	\includegraphics[width=0.48\textwidth]{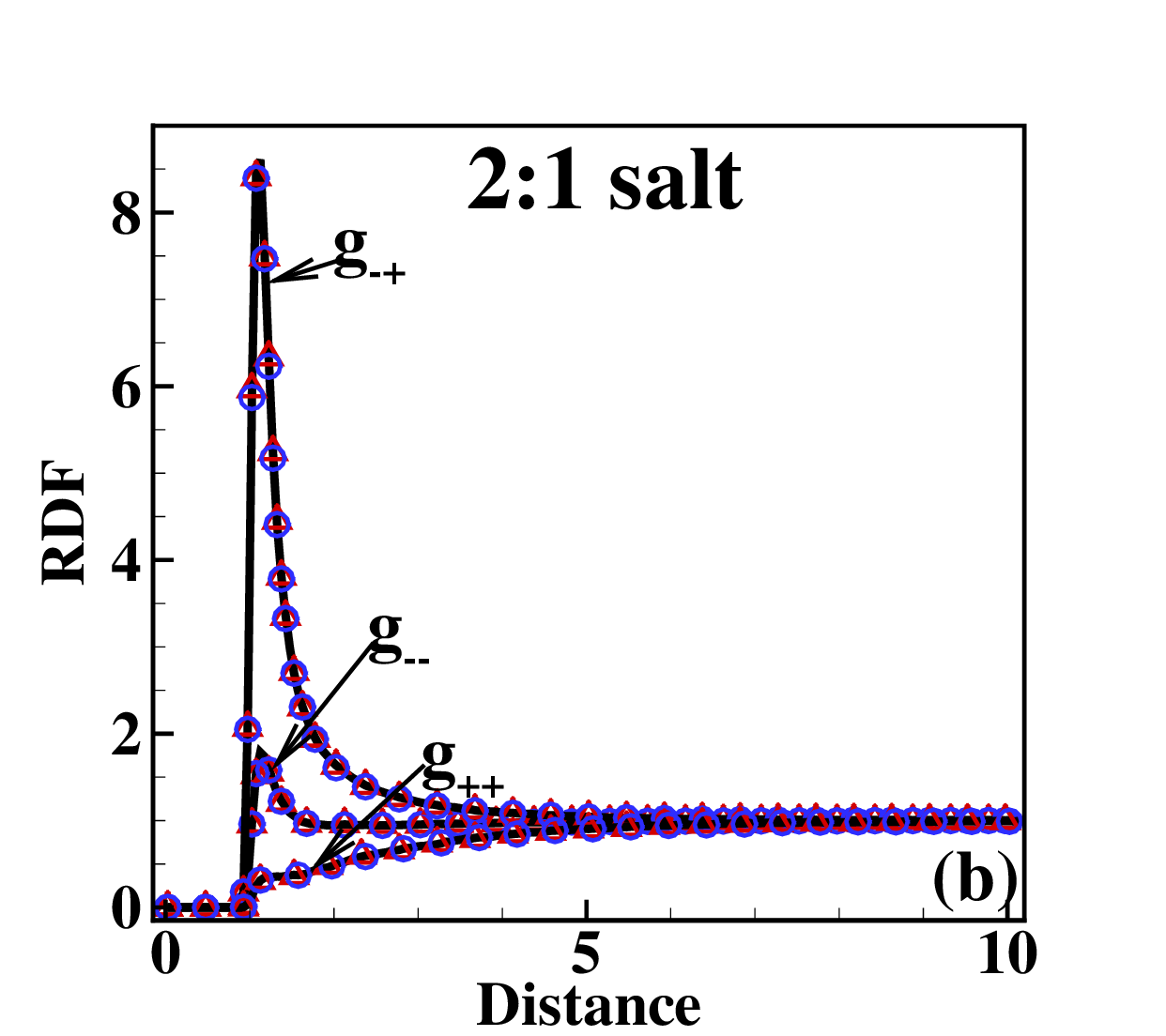}
	\includegraphics[width=0.48\textwidth]{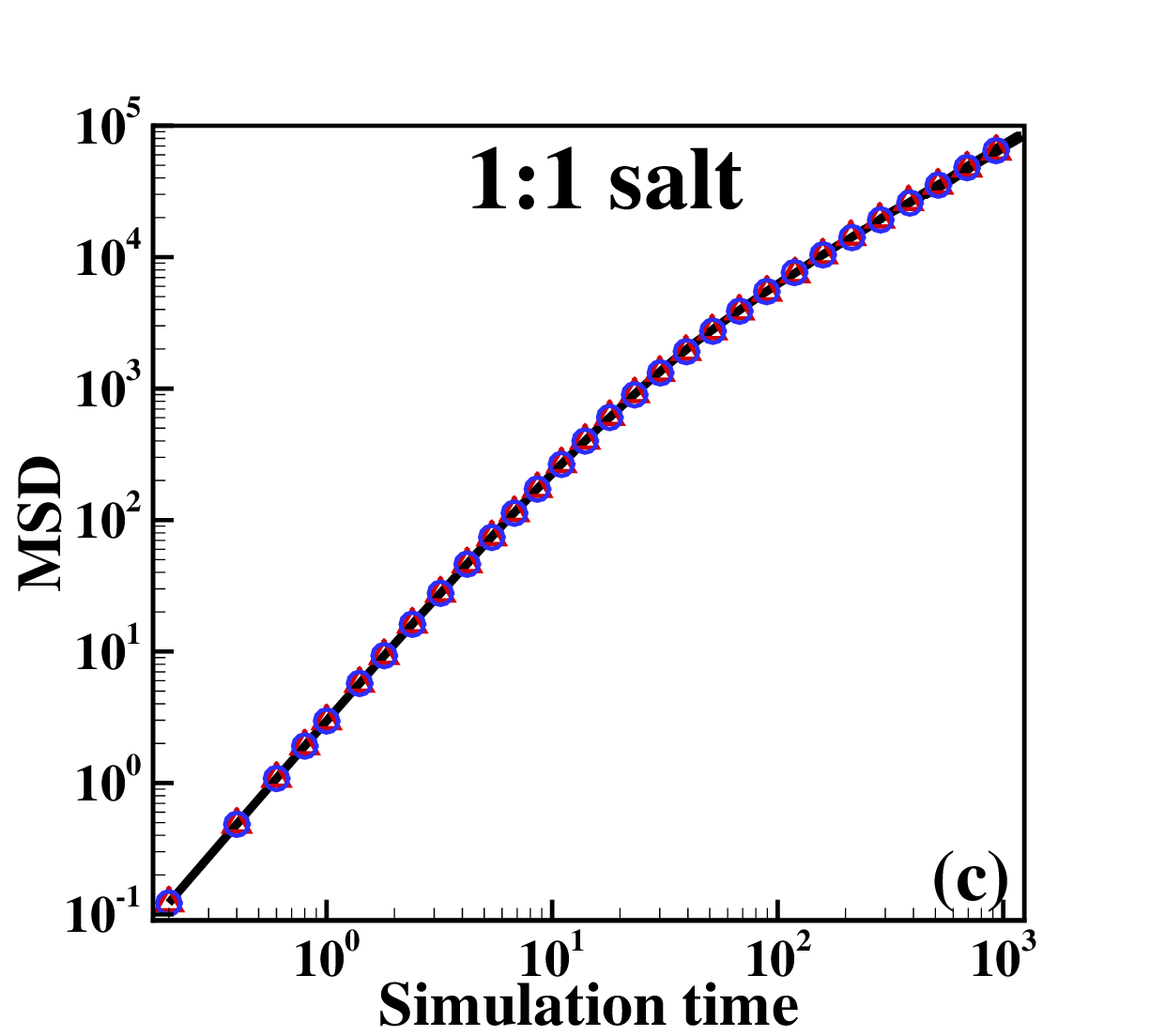}
	\includegraphics[width=0.48\textwidth]{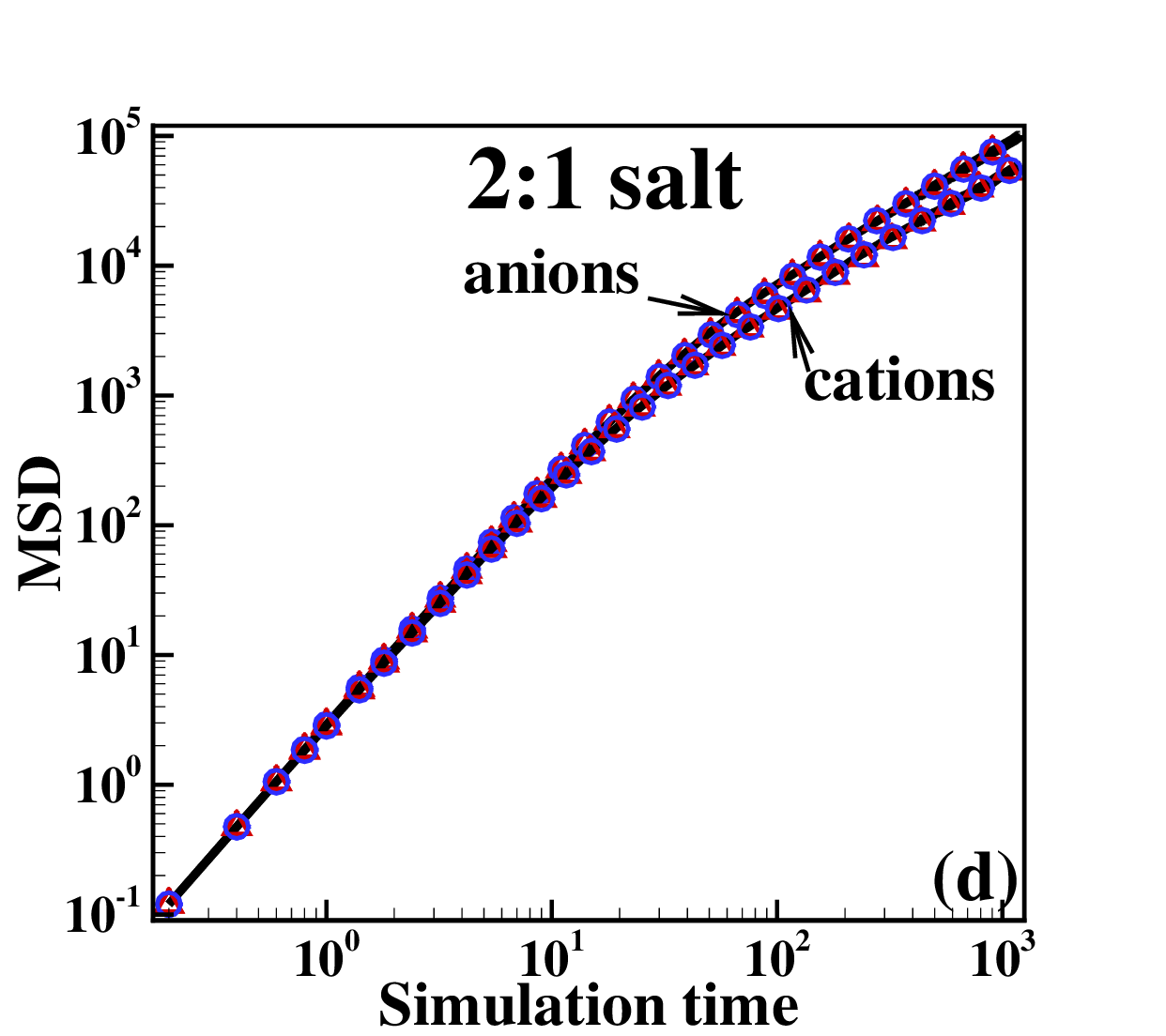}
	\caption{RDFs and MSDs of the 1:1 and 2:1 electrolytes calculated by three different methods.}
	\label{fig:rdfmsd}
\end{figure*}

\begin{figure*}[ht]	
	\centering
	\includegraphics[width=0.45\textwidth]{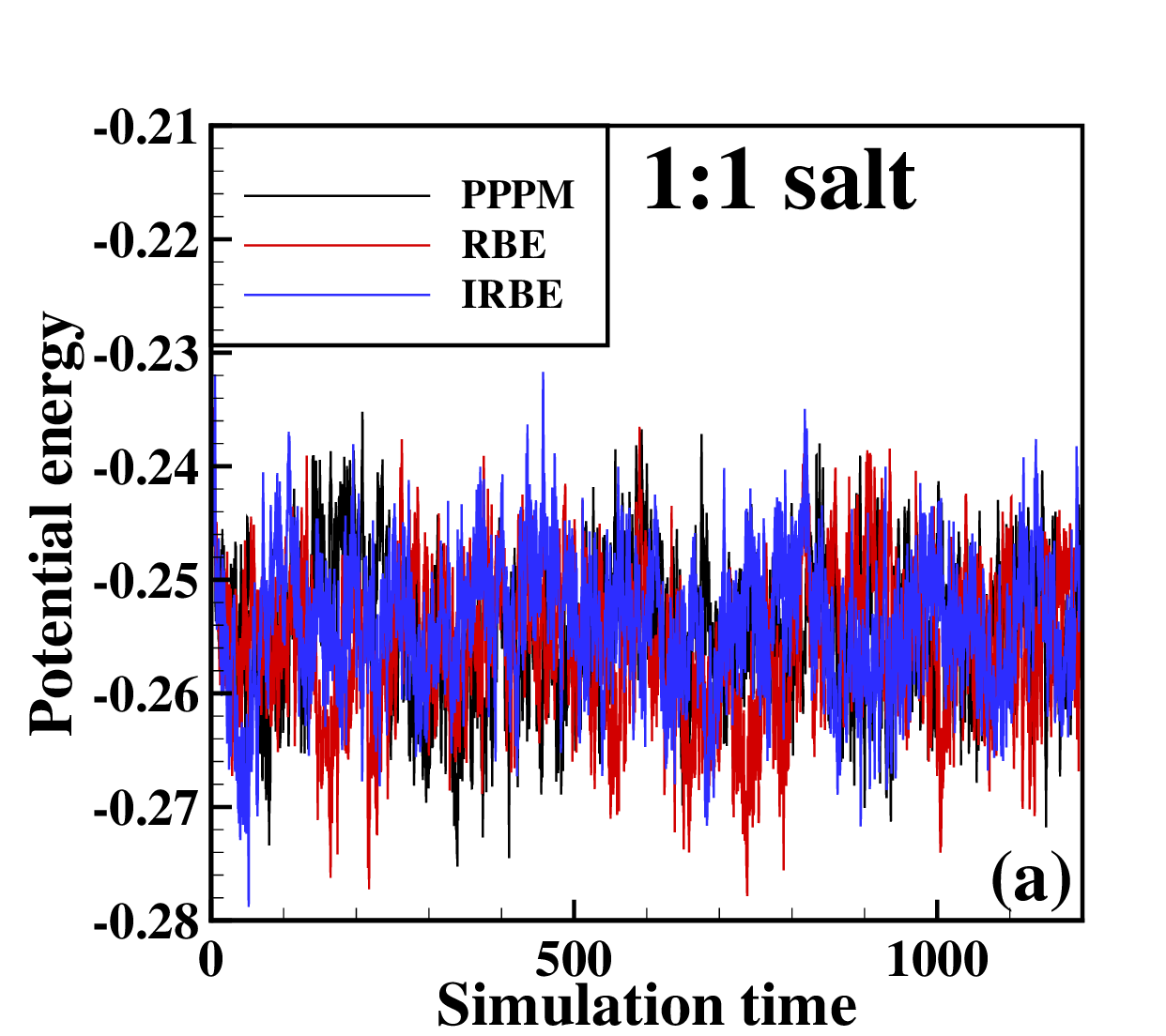}
	\includegraphics[width=0.45\textwidth]{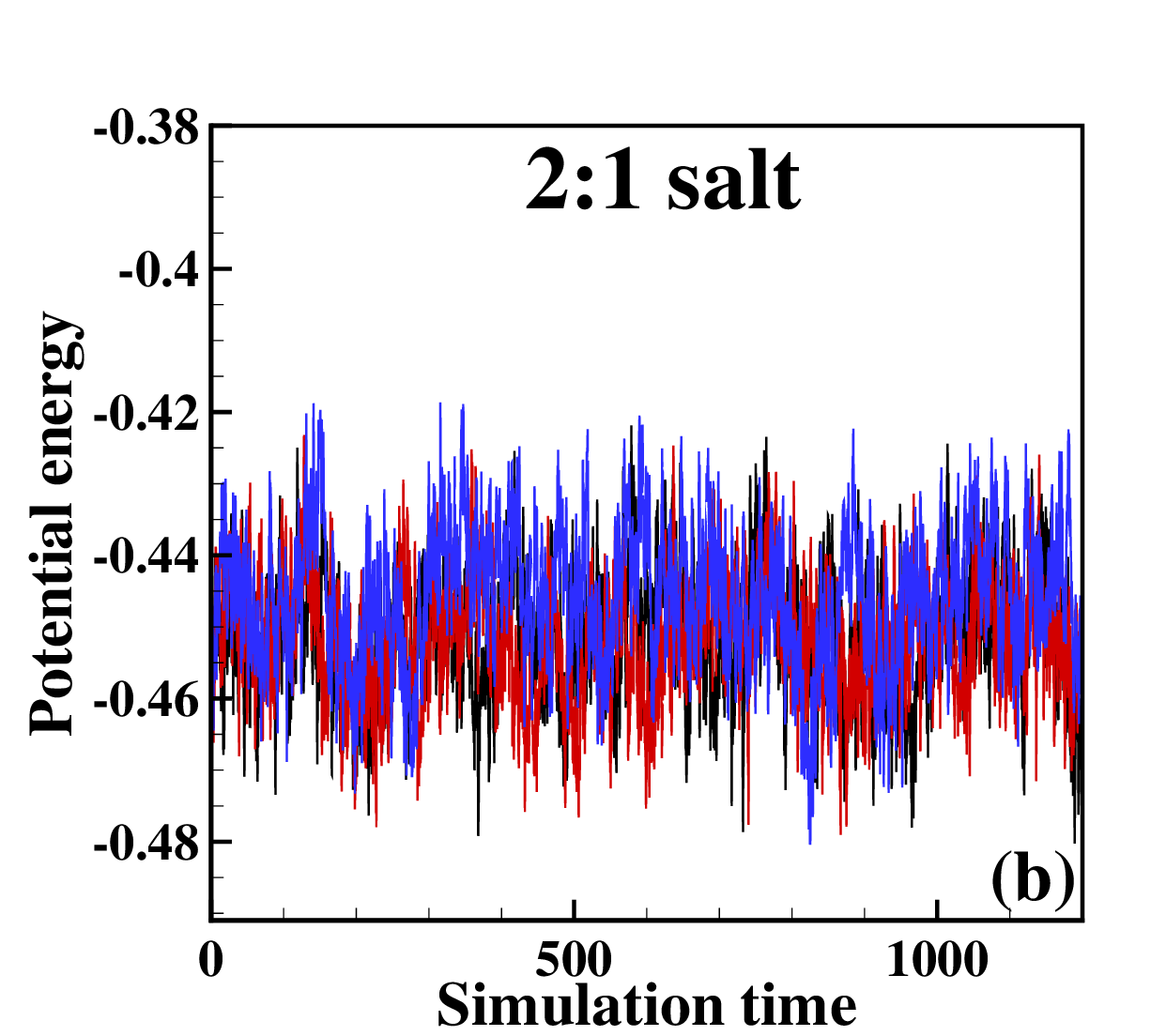}
	\caption{Potential energy per atom calculated by using the PPPM, RBE and the IRBE methods.}
	\label{fig:energy}
\end{figure*}

\renewcommand\arraystretch{1.4}
\begin{table}[ht]
	\centering
	\caption{Mean value (mean) and standard deviation (std) of the potential energy per atom displayed in Figure \ref{fig:energy}.}\label{tabl:ene}
	\setlength{\tabcolsep}{10mm}{
		\begin{tabular}{cccc}
			\hline\hline
			&PPPM&RBE&IRBE\\\hline
			$1:1$ mean&$-0.2544$&$-0.2559$&$-0.2538$\\\hline
			$1:1$ std&$0.005968$&$0.006289$&$0.005885$\\\hline
			$2:1$ mean&$-0.4505$&$-0.4516$&$-0.4467$\\\hline
			$2:1$ std&$0.008918$&$0.008393$&$0.009519$\\\hline\hline
	\end{tabular}}
\end{table}

We also test the CPU time of the algorithms by varying the system size while maintaining the particle density $N/L^3=0.01$.
We measure the time cost per step and the results are shown in Fig.~\ref{fig:timesalt},
where the number of particles takes from $6\times 10^5$ to $4\times 10^7$ by using $960$ CPU cores.
It can be observed that the simulation time of the IRBE has a linear scaling with the number of particles,
and the computational efficiency by using the IRBE method is improved by a factor of 2 compared to the RBE,
and a factor of 6 compared to the PPPM.
\begin{figure*}[ht]	
	\centering
	\includegraphics[width=0.45\textwidth]{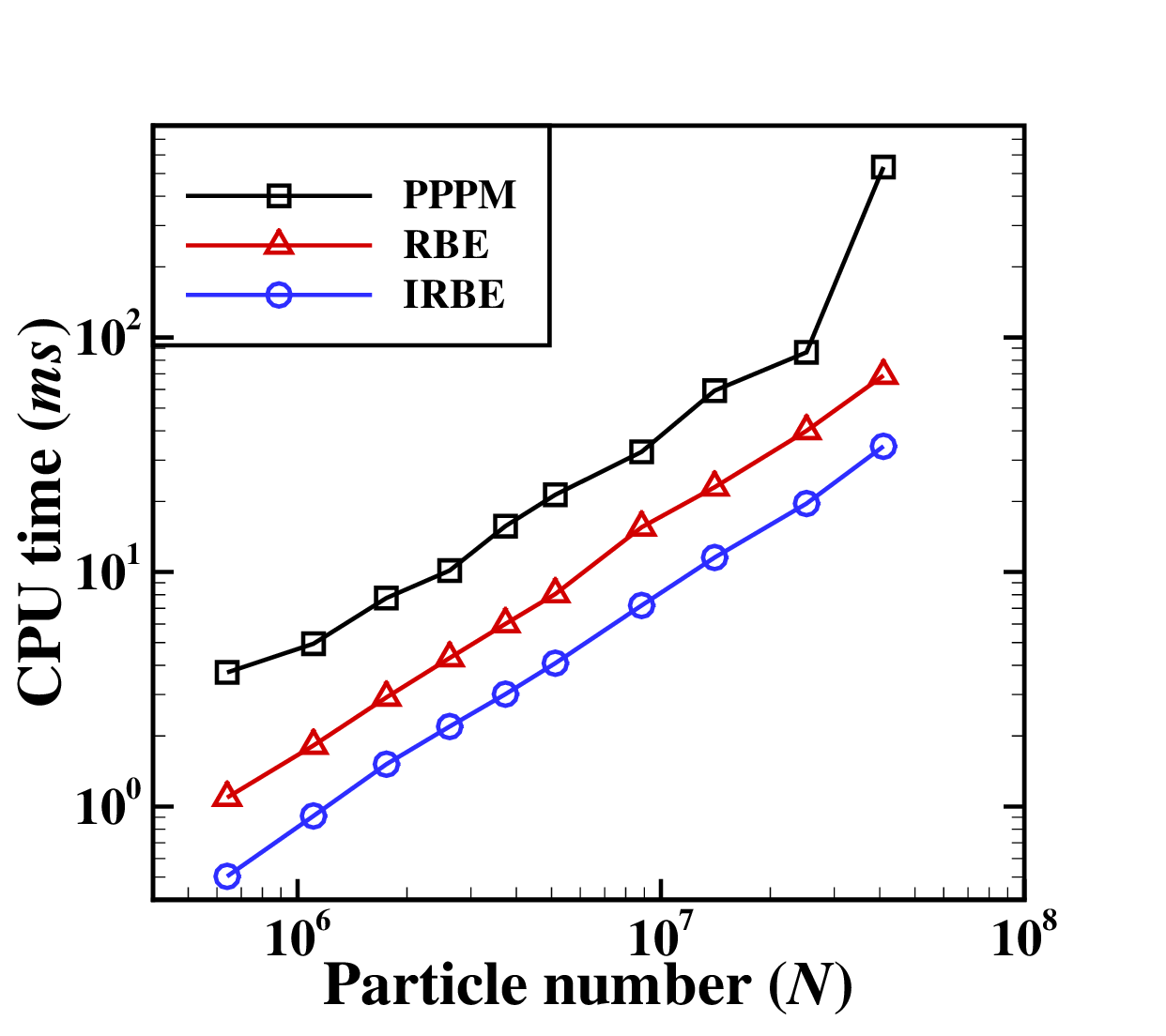}
	\caption{CPU time performance per timestep of the monovalent salt solution by using the PPPM, RBE and IRBE methods.}
	\label{fig:timesalt}
\end{figure*}

\subsection{All-atom water systems}
For the second benchmark problem, we conduct all-atom MD simulations of bulk water systems by
using the SPC/E force field\cite{mark2001structure}.
For validating the accuracy, a system consisting of $901$ water molecules is used with simulation box
of the side length $L=3 nm$. The Nos\'e-Hoover thermostat is applied with the coupling parameter $\gamma=1.0$ and temperature $T=298K$.
In each simulation, we perform $5\times 10^5$ time steps for equilibrium and another $6\times 10^5$ steps for statistics.
We choose the accuracy $10^{-4}$ in the PPPM method as reference.
For the IRBE, we set parameters $(r_\mathrm{\eta},\widetilde{p})=(0.7nm,100)$ for the short-range interactions,
and $(\alpha,p)=(0.09,200)$ for the long-range interactions. The setup balances the calculation time of these two parts and is near-optimal for time cost.
Table.~\ref{tabl:para} list other groups of parameters which can achieve similar accuracy (with the error of RDF, MSD and energy is in $<1\%$ error threshold) by our extensive simulations, which shall be useful to provide guidance for general all-atom simulations.
All the different parameter sets given in Table.~\ref{tabl:para} can produce accurate simulation results due to that variance of the stochastic force is controlled, demonstrating the insensitivity of the IRBE to these parameters.
With the increase of $\alpha$, the force variance in Fourier space becomes larger and we need to increase $p$.
At the same time, the error complementary function decreases, so one can choose smaller $(r_\mathrm{\eta},\widetilde{p})$. These parameters are determined empirically and are difficult to be given theoretically. Some valuable criteria may be related to the practical choice of parameters, e.g., the symmetry-preserving mean-field condition presented recently\cite{hu2022symmetry}. We will report further demonstrations of these criteria in our subsequent work. Fig.~\ref{fig:water} presents the RDFs and MSDs produced by different methods between oxygen atoms.
These data vadiate the accuracy of the IRBE method as it well reproduces the result of the other two methods.

\begin{figure*}[ht]	
	\centering
	\includegraphics[width=0.45\textwidth]{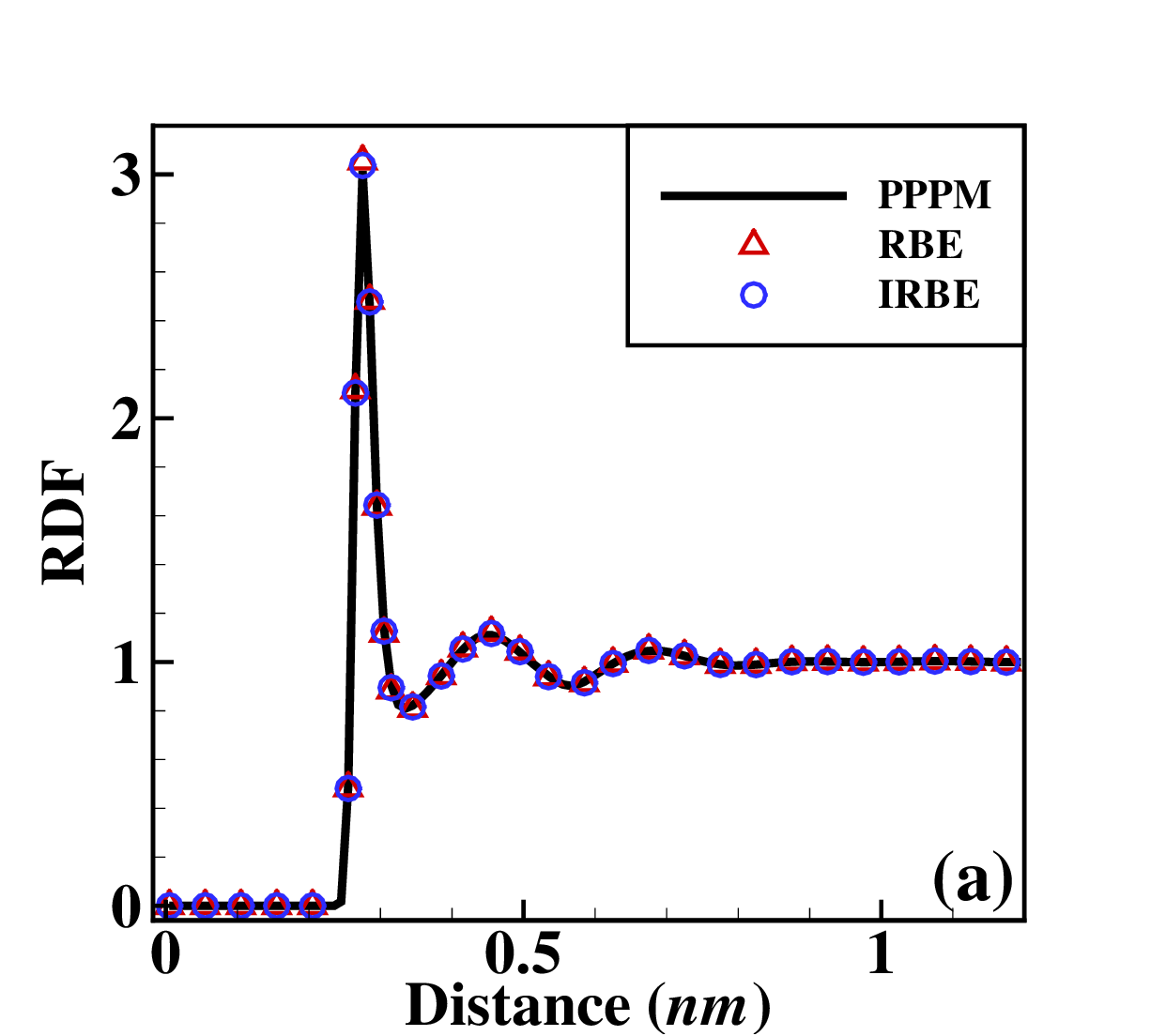}
	\includegraphics[width=0.45\textwidth]{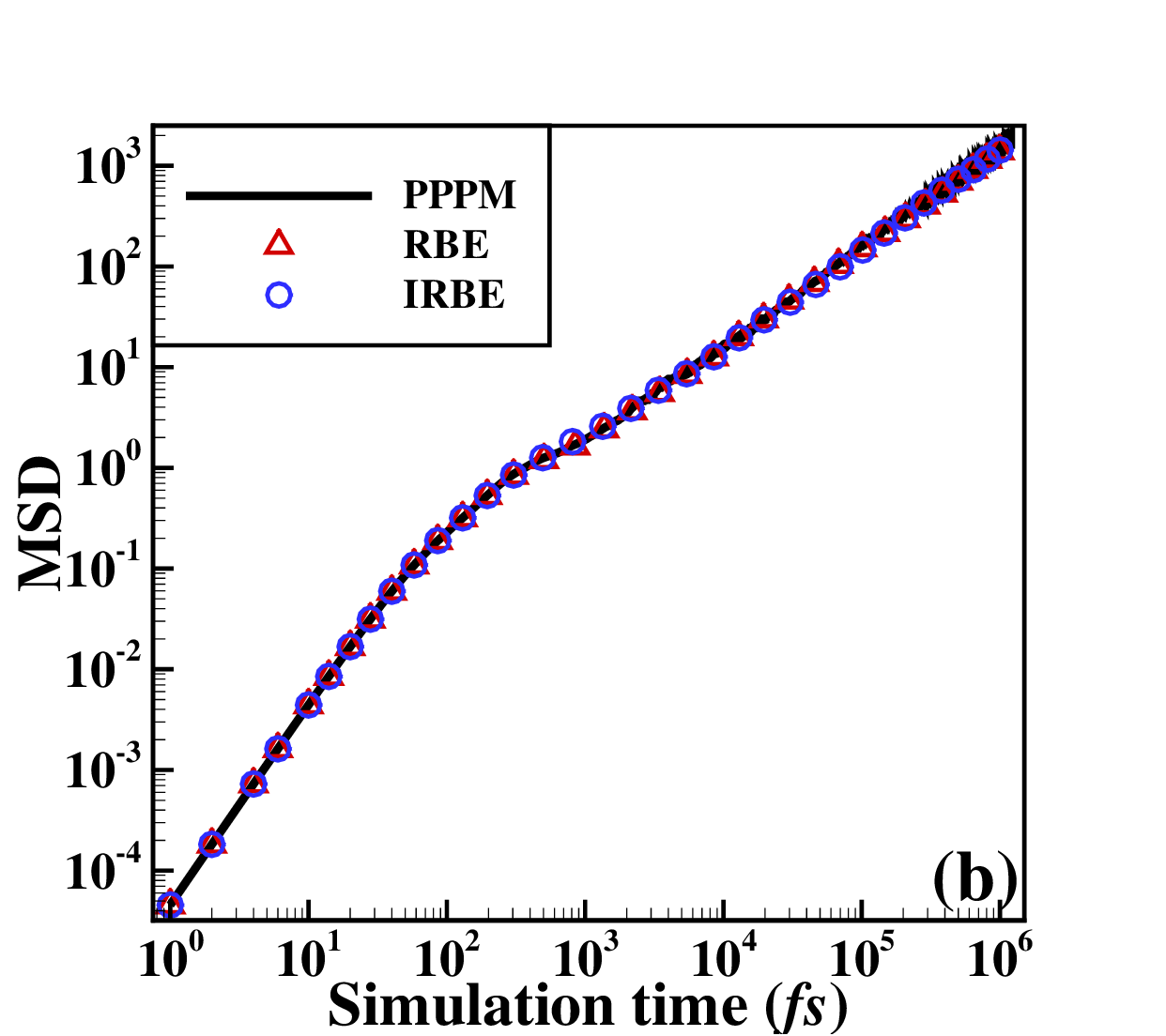}
	\caption{RDFs between oxygen atoms and MSDs of oxygen atom calculated by using the PPPM, RBE and IRBE methods for the bulk water system.}
	\label{fig:water}
\end{figure*}

\begin{table}[ht]
	\caption{The choice of parameters in the IRBE method for the bulk water system, where the result of
		RDF, MSD and energy is in $<1\%$ error threshold. (Unit of length: $nm$) }
	\centering
	\setlength{\tabcolsep}{10mm}{
		\begin{tabular}{ccc}
			\hline\hline
			\multicolumn{2}{c}{Fourier ~~space}& Real space ($r_\mathrm{c}=1.2nm$) \\
			$\alpha$ &$p$ &  $(r_\mathrm{\eta},\widetilde{p})$  \\\hline
			\midrule
			0.10&300&(0.6,~100) and (0.7,~~50)\\
			0.09&200&(0.6,~150) and (0.7,~100)\\
			0.08&200&(0.6,~200) and (0.7,~100)\\
			0.07&200&(0.6,~300) and (0.7,~200)\\
			0.06&100&(0.6,~400) and (0.7,~300)\\
			0.05&100&(0.6,~500) and (0.7,~400)\\
			\hline\hline
		\end{tabular}
	}
	\label{tabl:para}
\end{table}

We then study the CPU performance with the same algorithm parameters by taking a large size of water systems 
with $10^8$ water molecules and varying the number of CPU cores from 1000 to $50000$.
The results are shown in Fig.~\ref{fig:timewater}.
The CPU time per step of the IRBE method has linearly decrease with respect to the number of CPU cores,
while the parallel efficiency of the PPPM method decrease when the core number exceeds $10000$.
The decrease of the scalability of the PPPM can be understood as the increase of cost in CPU communication.
In comparison to the RBE, the IRBE method further improves the computational efficiency due to the acceleration
in calculating short-range interactions.
\begin{figure*}[ht]	
	\centering
	\includegraphics[width=0.45\textwidth]{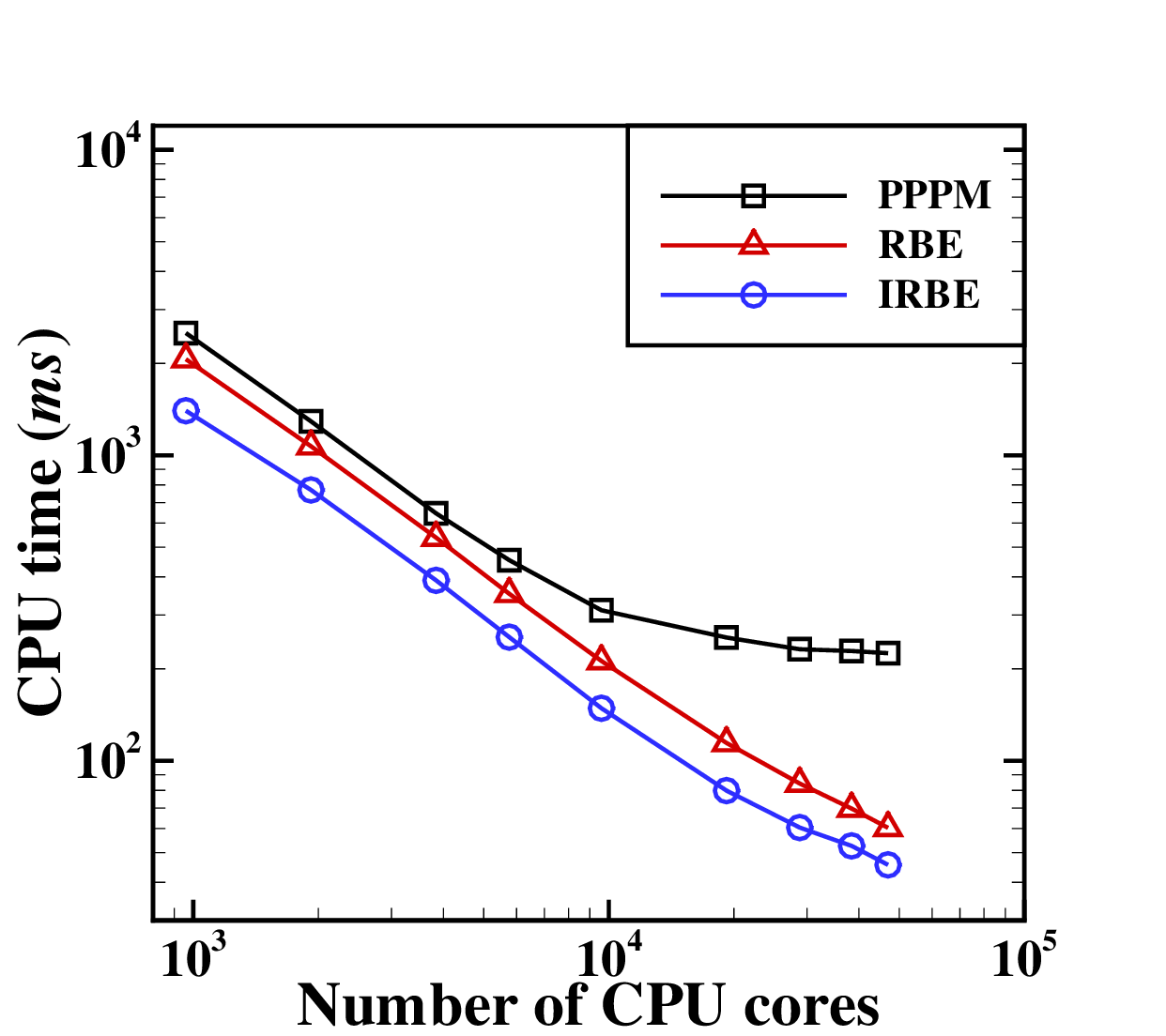}
	\caption{CPU cost per step of the bulk water system by using the PPPM, RBE and IRBE methods.}
	\label{fig:timewater}
\end{figure*}

\subsection{All-atom LiTFSI ionic liquid}
The third example is an all-atom system of a LiTFSI electrolyte of ultrahigh concentration ($5M$) with a cubic simulation box of initial size $5.67 nm$ including $320$ Li$^+$, $320$ TFSI$^{-}$, and $3561$ water molecules. This system has been previously studied \cite{liang2022superscalability} by the RBE. The high ionic concentration results in a strongly heterogeneous equilibrium distribution of the electrolyte due to the phase separtion (water versus anions) of length scale $1-2nm$. The system employs the optimized potential for all-atom force fields for Li$^+$  and  for TFSI$^-$, together with the TIP3P model for water molecules. The system is first equilibrated
in the NPT ensemble with the PPPM at $298$ $K$ and $1$ $bar$ for $500$ $ps$, followed by $500$ $ps$ production MD in the NVT using the Nos\'e-Hoover thermostat with the PPPM, RBE and IRBE, respectively. The cutoff radius of the short-range Coulomb interaction of the PPPM and
LJ interaction is $1.2nm$ with the splitting parameter $\alpha=0.05$, such that the relative error is about $10^{-4}$, which acts as the reference solution. Note that the result of the RBE method with $r_c=1.2 nm$ predicts the the same curve as the PPPM and the data is not present.

To compare the performance, we set up two simulations with parameters $(\alpha,p)=(0.05,500)$ and the short-range cutoff radius taking $0.8nm$ for the RBE and the IRBE, and the IRBE has the batch size $\widetilde{p}=100$ for the neighbor list. The sizes of neighbor lists for the RBE and IRBE are $\sim 200$ and $300$, respectively. The simulation results are present in Fig.~\ref{fig:LiTFSI}, where panel (a) illustrates the oxygen-oxygen RDFs and panel (b) displays the MSDs of Li$^+$ ions. 
It can be observed that the structure and the MSD features by the PPPM are well reproduced by the IRBE-based simulations, though a much smaller number of neighbors is used.
For comparison, the RDF curve by the RBE displays an obvious deviation at the peak, demonstrating that a direct use of smaller cutoff radius is less accurate. Consequenently, the MSD curve of the RBE is also gradually deviated from that of the PPPM with the increase of the simulation time, but the curve of the IRBE agrees well with the reference solution. It is noticed that the MSD of the lithium ions has a quadratic dependence on time before the phase separation, and the inflexion point at $50fs$ indicates that the system reaches the anomalous diffusive regime of the ions due to the correlation between dynamic heterogeneity and the local structural environment. This phenomenon is typical for ionic liquids at low temperature, and is captured by all three simulations.


\begin{figure*}[ht]	
	\centering
	\includegraphics[width=0.45\textwidth]{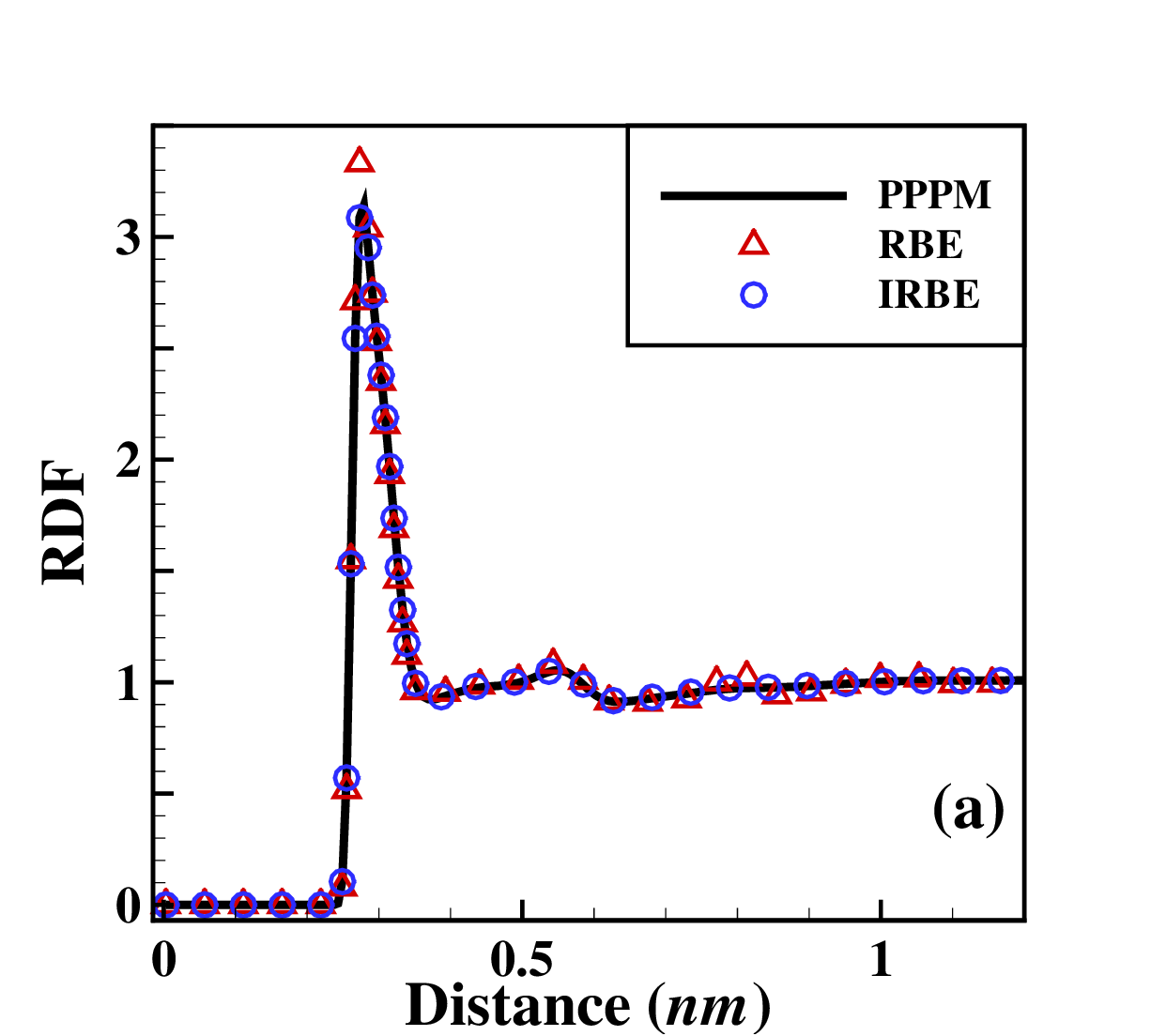}
	\includegraphics[width=0.45\textwidth]{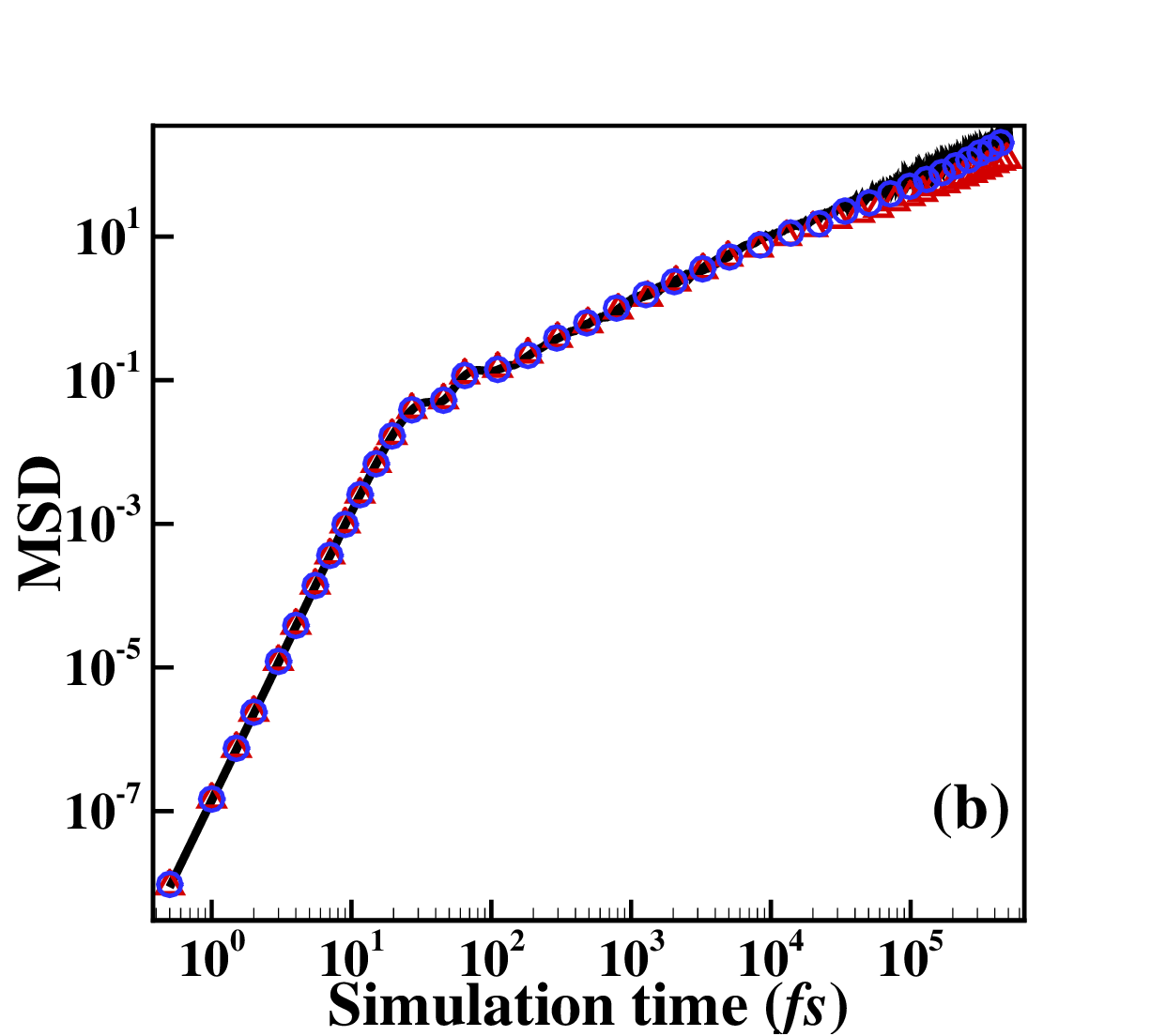}
	\caption{Oxegen-oxegen RDFs and MSDs of lithium ions calculated by MD simulations with the PPPM, RBE and IRBE methods for the LiTFSI electrolyte. The cutoff radius for the RBE and IRBE is $0.8 nm$, but the IRBE has additional 100 random neighbors for each atom.}
	\label{fig:LiTFSI}
\end{figure*}

\section{Conclusions}
In summary, we develop an improved version of the RBE method for non-bonded interactions in MD simulations,
which accurately and efficiently reproduces the structure and dynamical information of the RBE and the PPPM
method for benchmark problems.
The IRBE method is a stochastic approximation to the non-bonded forces, which benefits from random mini-batch
strategy in both the Fourier space for the long-range interactions and the real-space cutoff for the short-range
interactions, leading to efficient and memory-saving algorithm, achieving an optimal $\mathcal{O}(N)$ scaling.
Analysis on the stability and convergency are also provided.
The simulations on primitive-model electrolyte and all-atom bulk water systems are conducted to demonstrate the accuracy and attractive performance of the IRBE algorithm.

The IRBE method is implemented in parallel programming with the MPI and OpenMP. It will be also great significance
for broader use of the algorithm if there is a version with the graphic processing units, and we are working on
it currently. Moreover,
extension of the method to quasi-2D systems with planar interfaces, in particular with the dielectric mismatch \cite{liang2020harmonic,LIANG2022108332,yuan2021particle} is straightforward and shall be studied in our subsequent work. 
The IRBE method is expected to play an important role for simulating problems in applications
at chemical physics, materials and biological systems.

\section*{Acknowledgements}
The authors acknowledge the financial support from the National Natural Science Foundation of China (grant No. 12071288), Science and Technology Commission of Shanghai Municipality (grant Nos. 20JC1414100 and 21JC1403700), and the support from  the HPC center of Shanghai Jiao Tong University.

\section*{Conflict of interest}
The authors declare that they have no conflict of interest.

\section*{Data Availibility Statement}
The data that support the findings of this study are available from the corresponding author upon reasonable request.


\end{document}